\documentclass[leqno]{article}

\usepackage[]{amsfonts} \usepackage[]{amsmath}
\usepackage[]{amssymb} \usepackage[]{latexsym}
\usepackage{graphicx,color} \usepackage{amsthm}
\usepackage{mathrsfs} \usepackage{url}
\usepackage{cancel}
\usepackage{natbib}

\numberwithin{equation}{section}

\begin{document}

\newtheorem{teo}{Theorem}[section] \newtheorem*{teo*}{Theorem}
\newtheorem{prop}[teo]{Proposition} \newtheorem*{prop*}{Proposition}
\newtheorem{lema}[teo]{Lemma} \newtheorem*{lema*}{Lemma}
\newtheorem{cor}[teo]{Corollary} \newtheorem*{cor*}{Corollary}

\theoremstyle{definition}
\newtheorem{defi}[teo]{Definition} \newtheorem*{defi*}{Definition}
\newtheorem{exem}[teo]{Example} \newtheorem*{exem*}{Example}
\newtheorem{obs}[teo]{Remark} \newtheorem*{obs*}{Remark}
\newtheorem{ass}[teo]{Assumption}
\newtheorem*{hipo}{Hypotheses}
\newtheorem*{nota}{Notation}

\newcommand{\ds}{\displaystyle} \newcommand{\nl}{\newline}
\newcommand{\eps}{\varepsilon}
\newcommand{\LMMR}{\mbox{LMMR}}
\newcommand{\bE}{\mathbb{E}}
\newcommand{\cA}{\mathcal{A}}
\newcommand{\cB}{\mathcal{B}}
\newcommand{\cC}{\mathcal{C}}
\newcommand{\cL}{\mathcal{L}}
\newcommand{\cS}{\mathcal{S}}
\newcommand{\cO}{\mathcal{O}}
\newcommand{\cQ}{\mathcal{Q}}
\newcommand{\cH}{\mathcal{H}}
\newcommand{\cF}{\mathcal{F}}
\newcommand{\cM}{\mathcal{M}}
\newcommand{\cI}{\mathcal{I}}
\newcommand{\cD}{\mathcal{D}}
\newcommand{\cT}{\mathcal{T}}
\newcommand{\cP}{\mathcal{P}}
\newcommand{\bP}{\mathbb{P}}
\newcommand{\bT}{\mathbb{T}}
\newcommand{\bD}{\mathbb{D}}
\newcommand{\bQ}{\mathbb{Q}}
\newcommand{\bC}{\mathbb{C}}
\newcommand{\bN}{\mathbb{N}}
\newcommand{\bR}{\mathbb{R}}
\newcommand{\rhor}{\raisebox{1.5pt}{$\rho$}}
\newcommand{\varphir}{\raisebox{1.5pt}{$\varphi$}}
\newcommand{\taur}{\raisebox{1pt}{$\tau$}}
\newcommand{\bfP}{\mbox{\textbf{P}}}
\newcommand{\bfI}{\mbox{\textbf{I}}}

\title{First-Order Asymptotics of Path-Dependent Derivatives in Multiscale Stochastic Volatility Environment}

\author{Yuri F. Saporito}

\maketitle

\begin{abstract}
In this paper, we extend the first-order asymptotics analysis of Fouque \textit{et al.} to general path-dependent financial derivatives using Dupire's functional It\^o calculus. The main conclusion is that the market group parameters calibrated to vanilla options can be used to price to the same order exotic, path-dependent derivatives as well. Under general conditions, the first-order condition is represented by a conditional expectation that could be numerically evaluated. Moreover, if the path-dependence is not too severe, we are able to find path-dependent closed-form solutions equivalent to the fist-order approximation of path-independent options derived in Fouque \textit{et al.} Additionally, we exemplify the results with Asian options and options on quadratic variation.
\end{abstract}

\section{Introduction}



A natural generalization of the Black-Scholes model is within the framework of \textit{stochastic volatility models}. In these models, the volatility of the underlying asset is no longer assumed constant, but it is now modeled by a stochastic process.

Differently from the Black--Scholes model, there are virtually no closed-form solutions for option prices in stochastic volatility models, and hence it might be very difficult to get accurate option prices. An honorable exception is the quasi-closed solution of affine models as, for instance, the Heston model, \cite{heston93}.

The \textit{multiscale stochastic volatility models} of J.-P. Fouque, G. Papanicolaou, R. Sircar, and K. S{\o}lna are a powerful approach to reconcile stochastic volatility models and computational tractability of option prices (and calibration); see, for example, \cite{multiscale_fouque_new_book}.

The goal of this paper is to generalize the perturbation framework of Fouque \textit{et al.} to general path-dependent structures for the financial derivative payoff. The Functional It\^o Calculus, introduced by Bruno Dupire in the seminal paper \cite{fito_dupire}, is a tailored-made theory to handle path-dependence in It\^o's stochastic processes setting and hence it will be the chosen tool in our paper.

The main conclusion of our paper is that the first-order approximation in the path-dependent case is a straightforward generalization of the approximation in the classical, path-\textit{independent} case. Indeed, we have concluded that market group parameters are the same for the approximation of path-independent and path-dependent payoffs. This fact was verified for various path-dependent options previously, but in this paper we are able to prove, under mild smoothness assumptions, this result directly to any path-dependent structure. Furthermore, we show that the closed-form solutions for the first-order approximation found in the classical case is also established when the path-dependence is not too strong. Moreover, we consider Asian options and options on quadratic variation to exemplify the results.

In Section \ref{sec:mean_reversion}, we will provide the main results of the first-order perturbation analysis in the classical context of Fouque \textit{et al}. Then, in Section \ref{sec:fito_multiscale}, under the functional It\^o framework, we extend the fist-order correction to the case where the payoff of the financial derivative has a path-dependent structure.

\section{Multiscale Stochastic Volatility Models}\label{sec:mean_reversion}

Intuitively, mean-reversion indicates the return of a stochastic process to its long-run mean, when this mean exists. We will be mainly interested in the speed at which the process mean-reverts. The typical mean-reverting process to have in mind is the Ornstein-Uhlenbeck (OU) process:
\begin{align}
dy_t = \kappa (m - y_t)dt + \nu \sqrt{2\kappa} dw_t, \label{eq:ou}
\end{align}
where $\kappa > 0$ is mean-reversion rate, $m$ is the long-run mean, $\nu > 0$ is the volatility and $(w_t)_{t \geq 0}$ is a Brownian motion. The mean-reversion aspect of the OU process lies in its drift. Whenever $y_t > m$, the drift is negative and it pushes $y_t$ down towards the long-run mean $m$. The case when $y_t < m$ is similar. We can also see that the bigger the $\kappa$, the stronger is the mean-reversion. Typical sample paths are shown in Figure \ref{fig:fast_mean_rever_example}. In Finance, mean-reversion arises from the modeling of commodities, interest rate, volatility, currency exchange rates, etc.

\begin{figure}[!ht]
\centering
\includegraphics[width=0.8\linewidth]{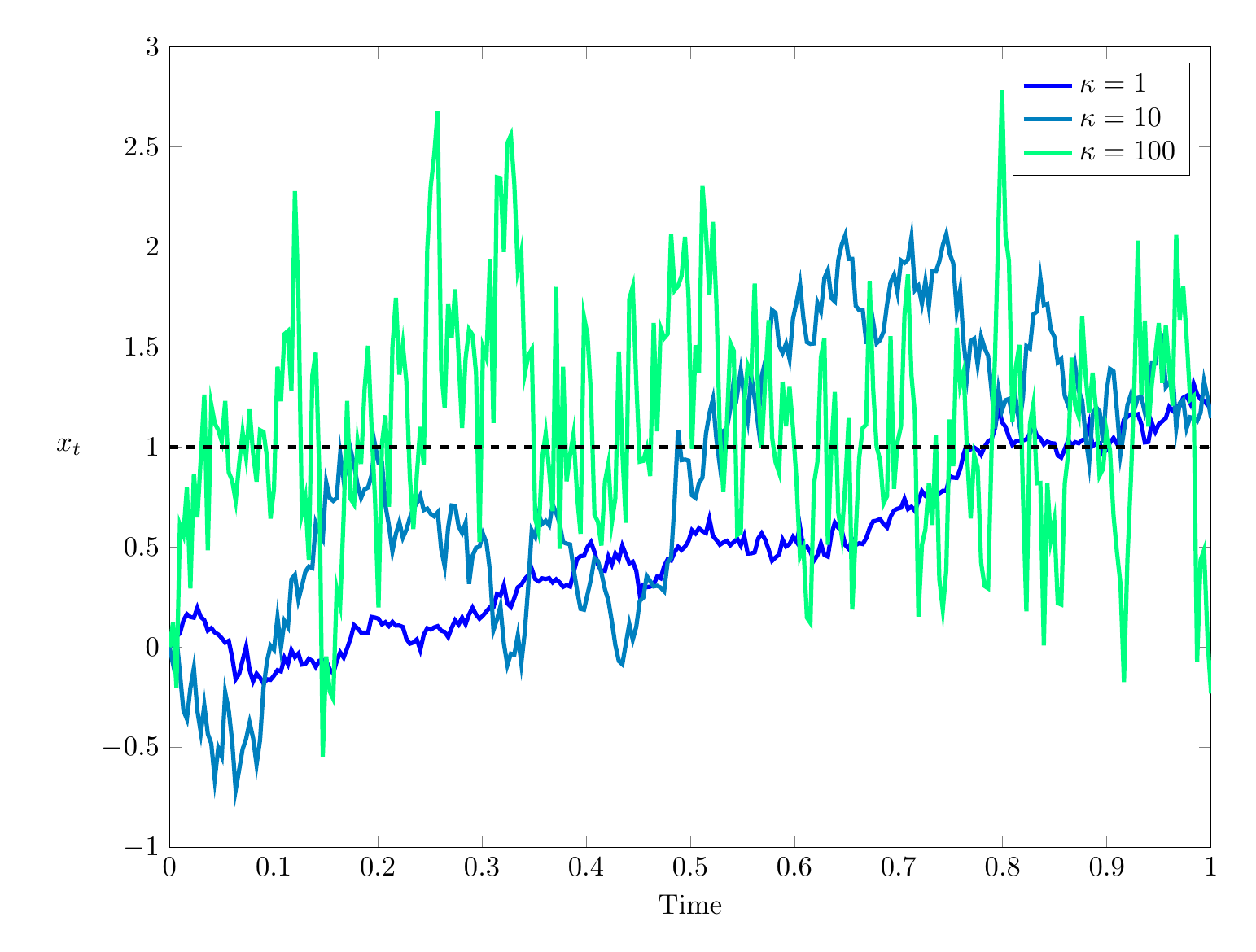}
\caption{Realizations of an OU process for different mean-reversion rates - Parameters: $x_0 = 0$, $m=1$, $\nu = 0.5$.}
\label{fig:fast_mean_rever_example}
\end{figure}


More formally, the notion of mean-reversion is expressed by the mathematically well-defined notion of ergodicity, see \cite[Section 3.2]{multiscale_fouque_new_book}. In our case, we will consider two processes, $y^\eps$ and $z^\delta$, that present fast and slow mean-reversion, respectively. Their dynamics can be written as follows
\begin{align}
&dy^{\eps}_t = \left(\frac{1}{\eps} \alpha(y_t^{\eps}) - \frac{1}{\sqrt{\eps}} \beta(y_t^{\eps}) \Gamma_1(y^{\eps}_t,z^{\delta}_t) \right)dt + \frac{1}{\sqrt{\eps}} \beta(y_t^{\eps}) dw_t^{(1)}, \\ 
&dz^{\delta}_t = \left(\delta c(y_t^{\eps}) - \sqrt{\delta} g(y_t^{\eps}) \Gamma_2(y^{\eps}_t,z^{\delta}_t) \right)dt + \sqrt{\delta} g(z^{\delta}_t) dw_t^{(2)}, \\
&dw_t^{(1)}dw_t^{(2)} = \rhor_{12} dt.
\end{align}
where $\alpha, \beta, c$ and $g$ satisfy certain requirements in order to guarantee the mean reversion of these processes.

If $T$ is the typical maturity of options contracts in this market, both $\eps$ and $\delta$ should be thought as small parameters in the sense that $\eps \ll T \ll 1/\delta$.

These two mean-reverting processes will model the time-scales of the volatility of the stock price we are modeling. More precisely, we will hereby assume that the stock price $s_t$, under a risk-neutral measure $\bQ$, follows
\begin{align}\label{eq:sde_risk_neutral_intro}
\left\{
\begin{array}{l}
ds_t = r s_t dt + f(y^{\eps}_t,z^{\delta}_t) s_t dw_t^{(0)},\\ \\
\ds dy^{\eps}_t = \left(\frac{1}{\eps} \alpha(y_t^{\eps}) - \frac{1}{\sqrt{\eps}} \beta(y_t^{\eps}) \Gamma_1(y^{\eps}_t,z^{\delta}_t) \right)dt + \frac{1}{\sqrt{\eps}} \beta(y_t^{\eps}) dw_t^{(1)}, \\ \\
dz^{\delta}_t = \left(\delta c(y_t^{\eps}) - \sqrt{\delta} g(y_t^{\eps}) \Gamma_2(y^{\eps}_t,z^{\delta}_t) \right)dt + \sqrt{\delta} g(z^{\delta}_t) dw_t^{(2)},
\end{array}
\right.
\end{align}
where $(w_t^{(0)},w_t^{(1)},w_t^{(2)})$ is a correlated $\bQ$-Brownian motion with
$$dw_t^{(0)}dw_t^{(i)} = \rhor_i dt, \ i=1,2, \ dw_t^{(1)}dw_t^{(2)} = \rhor_{12} dt.$$
The functions $\Gamma_1$ and $\Gamma_2$ together completely define the market price of volatility risk and uniquely determine the risk-neutral measure $\bQ$. The usual assumptions are required for the correlations and the functions $\Gamma_1$ and $\Gamma_2$, see Theorem \ref{thm:accuracy_intro}.

\subsection{First-Order Approximation}\label{sec:first_order}

This section will provide results on the first-order approximation in $\sqrt{\eps}$ and $\sqrt{\delta}$ of option prices when the volatility is governed by the dynamics described in Equation (\ref{eq:sde_risk_neutral_intro}).

Consider a European financial derivative with maturity $T$ and whose payoff $\varphi$ depends only on the terminal value of the stock, $s_T$, and hence called path-independent. The no-arbitrage price under $\bQ$ for this derivative is given by
$$P^{\eps,\delta}(t,x
,y,z) = \bE_{\bQ}[e^{-r(T-t)}\varphi(s_T) \ | \ s_t = x, y^{\eps}_t = y, z^{\delta}_t = z].$$
We are using the fact that $(s,y^{\eps},z^{\delta})$ is a Markov process. 

In Section \ref{sec:fito_multiscale}, we will perform the formal regular and singular perturbation analysis in the path-dependent framework. Here, we will list  the formulas of the approximation in the path-independent for the sake of comparison. Indeed, the formulas that we will find in the aforementioned section share the essence with the first-order approximation under this situation. For the complete analysis of the path-independent case, the reader is refereed to \cite{multiscale_fouque_new_book}. There the reader will also be able to examine the unfolding of this approach into diverse topics in Mathematical Finance.

Before proceeding, we will make precise the notation of our approximation results:
\begin{defi}
We say that a function $g^{\eps,\delta}$ is a first-order approximation in powers of $\sqrt{\eps}$ and $\sqrt{\delta}$ to the function $f^{\eps,\delta}$ if
$$|g^{\eps,\delta} - f^{\eps,\delta}| \leq C(\eps + \delta),$$
pointwise, for some constant $C > 0$ and for sufficiently small $\eps,\delta > 0$. We use the notation
\begin{align}
g^{\eps,\delta} - f^{\eps,\delta} = O(\eps + \delta). \label{eq:bigO_intro}
\end{align}
\end{defi}

We start the description of the first-order approximation by formally expanding $P^{\eps,\delta}$ is powers of $\sqrt{\eps}$ and $\sqrt{\delta}$:
\begin{align}
P^{\eps,\delta} = P_0 + \sqrt{\eps} P_{1,0} + \sqrt{\delta}P_{0,1} + \ldots \label{eq:first_order_P_intro}
\end{align}

Following the arguments presented in \cite{multiscale_fouque_new_book}, one can show that $P_0$, $P_{1,0}^{\eps} := \sqrt{\eps}P_{1,0}$ and $P_{0,1}^{\delta} := \sqrt{\delta}P_{0,1}$ should satisfy the following PDEs
\begin{align}
&\left\{
\begin{array}{l}
 \cL_{BS}(\bar{\sigma}(z))P_0(t,x,z) = 0, \\ \\
 P_0(T,x,z) = \varphi(x),
\end{array}
\right. \label{eq:pde_p0_intro} \\ \nonumber\\
&\left\{
\begin{array}{l}
 \cL_{BS}(\bar{\sigma}(z))P^{\eps}_{1,0}(t,x,z) = -\cA^{\eps}P_0(t,x,z), \\ \\
 P^{\eps}_{1,0}(T,x,z) = 0,
\end{array}
\right.\label{eq:pde_p_1_eps_intro}\\ \nonumber\\
&\left\{
\begin{array}{l}
 \cL_{BS}(\bar{\sigma}(z)) P_{0,1}^{\delta}(t,x,z) = -2\cA^{\delta}P_0(t,x,z), \\ \\
 P_{0,1}^{\delta}(T,x,z) = 0,
\end{array}
\right. \label{eq:pde_p_1_delta_intro}
\end{align}
where

\begin{align}
& \cL_{BS}(\sigma) = \frac{\partial}{\partial t} + \frac{1}{2} \sigma^2 D_2 + r(D_1 - \cdot),\\
&\bar{\sigma}^2(z) = \langle f^2(\cdot,z) \rangle, \label{eq:sigma_bar_intro} \\
&\cA^{\eps} = V_3^{\eps}(z) D_1 D_2 + V_2^{\eps}(z) D_2, \label{eq:cA_eps_intro} \\
&V_3^{\eps}(z) = -\frac{\rho_1\sqrt{\eps}}{2} \left\langle \beta \ f(\cdot,z) \frac{\partial \phi}{\partial y}(\cdot,z) \right\rangle, \label{eq:V3eps_intro} \\
&V_2^{\eps}(z) = \frac{ \sqrt{\eps}}{2} \left\langle \beta \ \Gamma_1(\cdot,z) \frac{\partial \phi}{\partial y}(\cdot,z) \right\rangle, \label{eq:V2eps_intro} \\
&\cA^{\delta} = -V_0^{\delta}(z) \frac{\partial}{\partial \sigma} - V_1^{\delta}(z) D_1 \frac{\partial}{\partial \sigma}, \label{eq:cA_delta_intro} \\
&V_0^{\delta}(z) = - \frac{g(z)\sqrt{\delta}}{2} \left\langle \Gamma_2(\cdot,z)\right\rangle \bar{\sigma}'(z), \label{eq:V0delta_intro} \\
&V_1^{\delta}(z) = \frac{\rho_2 g(z)\sqrt{\delta}}{2} \left\langle f(\cdot,z) \right\rangle \bar{\sigma}'(z),\label{eq:V1delta_intro}\\
&D_k = x^k \frac{\partial^k }{\partial x^k}. \label{eq:Dk_intro}
\end{align}

The function $\phi$ above is defined as the solution of the Poisson equation:
\begin{align}
\cL_0 \phi(y,z) = f^2(y,z) - \bar{\sigma}^2(z), \label{eq:poisson_eq_intro}
\end{align}
with $z$ being just a parameter, where $\cL_0$ is the infinitesimal generator of the process $y^1$ under the physical measure $\bP$, see Equation (\ref{eq:path_cL_0}).

It is worth noticing that the first-order approximation was chosen independently of $y$, the initial value of the process $y^{\eps}$. This is an important feature of this approximation because the process $y^{\eps}$ is unobservable and hence the estimation of $y$ would be difficult. Moreover, the dependence with respect to $z$, the initial value of $z^{\delta}$, is only through the parameters $(\bar{\sigma}(z), V_0^{\delta}(z), V_1^{\delta}(z), V_2^{\eps}(z), V_3^{\eps}(z))$. Therefore, it is not necessary to estimate the particular value of $z$ either.

One can further show the following explicit formulas are valid
\begin{align}
&P_0(t,x,z) = P_{BS}(t,x; \bar{\sigma}(z)), \label{eq:p0_into} \\
&P^{\eps}_{1,0}(t,x,z) = (T - t)\cA^{\eps} P_{BS}(t,x; \bar{\sigma}(z)), \label{eq:p10eps_into}\\
&P^{\delta}_{0,1}(t,x,z) = (T - t)\cA^{\delta} P_{BS}(t,x; \bar{\sigma}(z)), \label{eq:p01delta_into}
\end{align}
where $P_{BS}(t,x; \sigma)$ is the price at $(t,x)$ of the European option with maturity $T$ and payoff function $\varphi$ under the Black--Scholes model with constant volatility $\sigma$. Therefore, the leading term of the approximation is the Black--Scholes price of the option with the effective volatility $\bar{\sigma}(z)$ and the first-order correction only involves Greeks of this price. The proof of these representations rely heavily on the commutation of the undiscounted Black--Scholes PDE operator, $\cL_{BS}(\sigma) + r \cdot$, and the operators $\cA^\eps$ and $\cA^\delta$. This observation will be very important when considering path-dependent payoff.

The accuracy of this approximation can be proved under the following assumptions. For the proof, we forward the reader to \cite{fouque_second_order}.

\begin{teo}\label{thm:accuracy_intro}

We assume

\begin{enumerate}

\item Existence and uniqueness of the SDE (\ref{eq:sde_risk_neutral_intro}) for any fixed $(\eps,\delta)$;

\item The function $f(y,z)$ is measurable, bounded, bounded away from zero, smooth in $z$ and such the solution $\phi$ to the Poisson equation (\ref{eq:poisson_eq_intro}) is at most polynomially growing;

\item The process $y^1$ has a unique invariant distribution, is mean-reverting as defined in \cite[Section 3.2]{multiscale_fouque_new_book}, and has moments of any order uniformly in $t > 0$;

\item The process $z^1$ has moments of any order uniformly in $t \leq T$, for any fixed $T > 0$;

\item The market prices of volatility risk $\Gamma_1$ and $\Gamma_2$ are bounded;

\item The payoff function $\varphi$ is measurable, locally bounded and at most polynomially growing at 0 and $\infty$.

\end{enumerate}
Then,
$$P^{\eps,\delta}(t,x,y,z) = P_0(t,x,z) + P_{1,0}^{\eps}(t,x,z) + P_{0,1}^{\delta}(t,x,z) + O(\eps + \delta).$$
\end{teo}

A valuable feature of the perturbation method is that in order to compute the first-order approximation, we only need the values of the \textit{group market parameters}
$$(\bar{\sigma}(z), V_0^{\delta}(z), V_1^{\delta}(z), V_2^{\eps}(z), V_3^{\eps}(z)).$$
This feature can also be seen as model independence and robustness of this approximation: under the regularity conditions stated in Theorem \ref{thm:accuracy_intro}, this approximation is independent of the particular form of the coefficients describing the process $y^{\eps}$ and $z^{\delta}$, i.e. the functions $\alpha$, $\beta$, $c$ and $g$ involved in the model (\ref{eq:sde_risk_neutral_intro}). The group market parameters can be interpreted as follows
\begin{itemize}

\item $\bar{\sigma}(z)$ is the effective volatility;

\item $V_0^{\delta}(z)$ measures the first-order impact of part of the market price of volatility risk;

\item $V_1^{\delta}(z)$ has the same sign as the correlation of the slow factor and the stock price;

\item $V_2^{\eps}(z)$ measures the first-order impact of part of the market price of volatility risk;

\item $V_3^{\eps}(z)$ has the same sign as the correlation of the fast factor and the stock price.

\end{itemize}

\begin{obs}[Parameter Reduction]\label{rmk:parameter_red}

$V_2^{\eps}(z)$ can be incorporated into the effective volatility. More precisely, we may consider the corrected volatility level $\sigma^{\star}(z)$ defined as
\begin{align}
\sigma^{\star}(z) = \sqrt{\bar{\sigma}^2(z) + 2V_2^{\eps}(z)}. \label{eq:sigma_star_intro}
\end{align}
Using this new volatility level, one could show that
\begin{align}
&\overline{P}^{\eps,\delta}(t,x,z) = P_{BS}^{\star}(t,x,z) \label{eq:approx_reduced_intro} \\
&+ (T-t) \left(V_0^{\delta}(z) \frac{\partial P_{BS}^{\star}}{\partial \sigma} + V_1^{\delta}(z) D_1 \frac{\partial P_{BS}^{\star}}{\partial \sigma} +V_3^{\eps}(z) D_1 D_2 P_{BS}^{\star} \right), \nonumber
\end{align}
approximates $P^{\eps,\delta}$ to the first-order, where $P_{BS}^{\star}(t,x,z) = P_{BS}(t,x; \sigma^{\star}(z))$.

\end{obs}

\subsection{Calibration to Implied Volatilities}\label{sec:calibration_implied}

In terms of implied volatility, this perturbation analysis translates into an affine approximation in the log-moneyness to maturity ratio (LMMR), which is formally defined in Equation (\ref{eq:LMMR_intro}) below. 

One can show then that the first-order approximation of the implied volatility is
\begin{align}
b^{\star} + (T-t)b^{\delta} + (a^{\eps} + (T-t) a^{\delta}) \LMMR, \label{eq:implied_vol_approx_intro}
\end{align}
where
\begin{align}
&\LMMR = \frac{\log(K/x)}{T-t}, \label{eq:LMMR_intro} \\
&b^{\star} = \sigma^{\star}(z) + \frac{V_3^{\eps}(z)}{2\sigma^{\star}(z)} \left(1 - \frac{2r}{{\sigma^{\star}}^2(z)} \right), \ a^{\eps} = \frac{V_3^{\eps}(z)}{{\sigma^{\star}}^3(z)}, \label{eq:bstar_aeps_intro} \\
&b^{\delta} = V_0^{\delta}(z) + \frac{V_1^{\delta}(z)}{2} \left(1 - \frac{2r}{{\sigma^{\star}}^2(z)} \right), \ a^{\delta} = \frac{V_1^{\delta}(z)}{{\sigma^{\star}}^2(z)}. \label{eq:bdelta_aeps_intro}
\end{align}

Inverting the formulas (\ref{eq:bstar_aeps_intro}) and (\ref{eq:bdelta_aeps_intro}) to the first-order of accuracy, we find the calibration formulas
\begin{align}
&\sigma^{\star}(z) = b^{\star} + a^{\eps} \left(r - \frac{{b^{\star}}^2}{2} \right), \ V_3^{\eps}(z) = a^{\eps} {b^{\star}}^3, \label{eq:cal_for_1_intro} \\
&V_0^{\delta}(z) = b^{\delta} + a^{\delta} \left(r - \frac{{b^{\star}}^2}{2} \right), \ V_1^{\delta}(z) = a^{\delta} {b^{\star}}^2. \label{eq:cal_for_2_intro}
\end{align}
Therefore, one could very easily calibrate the parameters $(b^{\star}, b^{\delta}, a^{\eps}, a^{\delta})$ to real data and use the formulas above to compute the calibrated values of $(\sigma^{\star}(z), V_3^{\eps}(z),$ $V_0^{\delta}(z), V_1^{\delta}(z))$. Below, we present an example of the first-order approximation of an implied volatility surface.

\begin{figure}[!ht]
\centering
\includegraphics[width=0.8\linewidth]{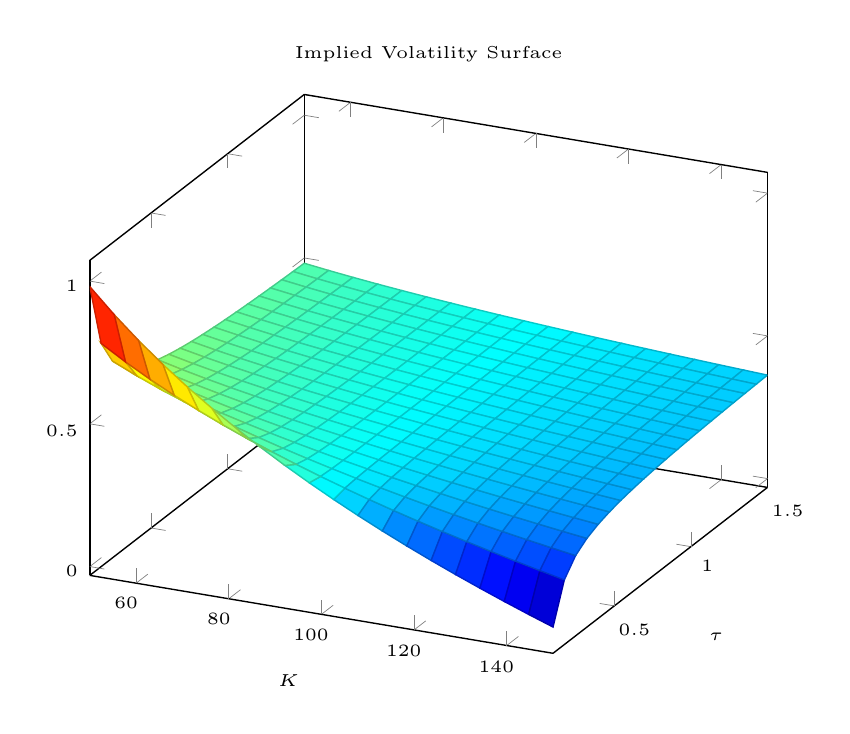}
\caption{First-Order Approximation of the Implied Volatility Surface - Parameters: $\sigma^{\star} = 0.4$, $V_0^{\delta} = 0.006$, $V_1^{\delta} = -0.009$, $V_3^{\eps} = -0.005$, $r = 0.05$.}
\label{fig:imp_vol_approx}
\end{figure}

\section{Path-Dependent Financial Derivatives}\label{sec:fito_multiscale}

Firstly, we will introduce the main notation, definitions and results of the Functional It\^o Calculus theory, as it was introduced in \cite{fito_dupire}, that will be necessary in what follows.

\subsection{A Brief Introduction to Functional It\^o Calculus}

The space of c\`adl\`ag paths up to time $t$ will be denoted by $\Lambda_t$. We also fix a time horizon $T > 0$. The \textit{space of paths} is then defined as
$$\Lambda = \bigcup_{t \in [0,T]} \Lambda_t.$$

We will denote elements of $\Lambda$ by upper case letters and the final time of its domain will be subscripted, e.g. $X \in \Lambda_t \subset \Lambda$ will be denoted by $X_t$. The value of $X_t \in \Lambda$ at a specific time will be denoted by lower case letter: $x_s := X_t(s)$, for any $s \leq t$. Moreover, if a path $X_t \in \Lambda$ is fixed, the path $X_s$, for $s \leq t$, will denote the restriction of the path $X_t$ to the set $[0,s]$. A \textit{functional} is any function $f:\Lambda \longrightarrow \bR$. The functional time and space derivatives are defined as the following limits, when they exist,
\begin{align}
\Delta_t f(X_t) &= \lim_{\delta t \to 0^+} \frac{f(X_{t,\delta t}) - f(X_t)}{\delta t}, \label{time_deriv}\\
\Delta_x f(X_t) &= \lim_{h \to 0} \frac{f(X_t^h) - f(X_t)}{h}, \label{space_deriv}
\end{align}
where $X_{t,\delta t}$ and $X_t^h$, for $\delta t > 0$ and $h \in \bR$, are given by
\begin{align*}
X_{t,\delta t}(u) &= \left\{
\begin{array}{l}
  x_u \ , \quad \mbox{if} \quad 0 \leq u \leq t, \\
  x_t \ , \quad \mbox{if} \quad t \leq u \leq t + \delta t,
\end{array}
\right.\\ \\
X_t^h(u) &= \left\{
\begin{array}{l}
  x_u \ , \quad \mbox{if} \quad 0 \leq u < t, \\
  x_t + h \ , \quad \mbox{if} \quad u = t,
\end{array}
\right.
\end{align*}
see Figures \ref{fig_flat} and \ref{fig_bump}.
\begin{figure}[!ht]
  \begin{minipage}[b]{0.5\linewidth}
    \centering
    \includegraphics[width=0.5\linewidth]{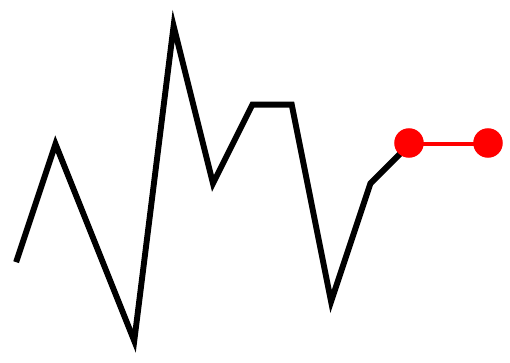}
    \caption{Flat extension of a path.}
    \label{fig_flat}
  \end{minipage}
  \begin{minipage}[b]{0.5\linewidth}
    \centering
    \includegraphics[width=0.5\linewidth]{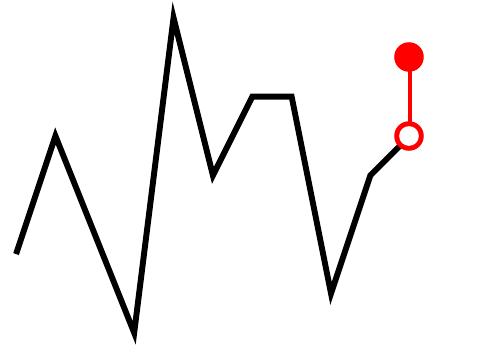}
    \caption{Bumped path.}
    \label{fig_bump}
  \end{minipage}
\end{figure}

For any $X_t, Y_u \in \Lambda$, where it is assumed without loss of generality that $u \geq t$, we consider the following metric in $\Lambda$:
$$d_{\Lambda}(X_t,Y_u) = \| X_{t,u-t} - Y_u\|_{\infty} + u -t.$$
Moreover, a functional $f$ is said $\Lambda$-continuous if it is continuous with respect to the metric $d_{\Lambda}$. Finally, a functional $f: \Lambda \longrightarrow \bR$ is said to belong to $\bC^{1,2}$ if it is $\Lambda$-continuous and it has $\Lambda$-continuous derivatives $\Delta_t f$, $\Delta_x f$ and $\Delta_{xx} f$.

Before continuing, we fix a probability space $(\Omega, \cF, \bP)$ and provide some comments about conditional expectation in the context of paths and functionals. For any $u \leq t$ in $[0,T]$, denote by $\Lambda_{u,t}$ the space of c\`adl\`ag paths in $[u,t]$. Now define the operator $(\cdot \ \otimes \ \cdot) : \Lambda_{u,t} \times \Lambda_{t,T} \longrightarrow \Lambda_{u,T}$, the \textit{concatenation} of paths, by
$$(X \otimes Y)(u) = \left\{
\begin{array}{ll}
  x_r \ , \quad &\mbox{if} \quad u \leq r \leq t \\
  y_r - y_t + x_t  \ , \quad &\mbox{if} \quad t \leq r \leq T,
\end{array}
\right.$$
which is a continuous paste of $X$ and $Y$. Let us consider a process $s$ given by the following (Markovian) Stochastic Differential Equation (SDE)
\begin{align}
ds_u = a(u,s_u)du + b(u,s_u)dw_u, \label{sde}
\end{align}
with $u \geq t$ and $s_t = x$. The unique strong solution of this SDE will be denoted by $s_u^{t,x}$ and the path solution from $t$ to $T$ by $S_T^{t,x}$. Finally, we define the \textit{conditioned expectation} as
\begin{align}
\bE[g(S_T) \ | \ X_t] = \bE[g(X_t \otimes S_{T}^{t, x_t})], \label{conditioned_expec}
\end{align}
for any $X_t \in \Lambda$. One could further show that $\bE[g(S_T) \ | \ S_t] = \bE[g(S_T) \ | \ \cF_t^s]$, $\bP$-a.s, where $\cF_t^s$ is the filtration generated by $s$.

\begin{ass}[Smoothness]\label{obs:smoothness}

We will assume henceforth that every functional considered in this paper is $\Lambda$-continuous and has $\Lambda$-continuous functional derivatives of all orders. This condition could be weakened, but it is outside the scope of this paper. The goal is to focus on the essential arguments that the functional framework brings.

\end{ass}

\begin{obs}\label{obs:result_delta_x}
We will frequently use the following result: if a functional $f \in \bC^{1,2}$ satisfies $f = 0$ for all continuous paths, then $\Delta_x f = 0$ for all continuous paths as well. That is, the functional space derivative on continuous paths of a $\bC^{1,2}$ functional is completely defined by its values on continuous paths. See, for instance, \cite{fito_fournie_thesis}[Theorem 2.2]. 
\end{obs}

The main results we will use in this paper are the following functional extensions of the It\^o and the Feynman-Kac Formulas:
\begin{teo}[Functional It\^o Formula]\label{fif}
Let $s$ be a continuous semi-martingale and $f \in \bC^{1,2}$. Then, for any $t \in [0,T]$,
$$f(S_t) = f(S_0) + \int_0^t \Delta_t f(S_u) du + \int_0^t \Delta_x f(S_u) ds_u + \frac{1}{2} \int_0^t \Delta_{xx} f(S_u) d\langle s\rangle_u.$$
\end{teo}

\begin{teo}[Functional Feynman-Kac Formula]\label{feynman-kac}
Let $s$ be a process given by the SDE (\ref{sde}). Consider the functionals $g :\Lambda_T \longrightarrow \bR$ and $k:\Lambda \longrightarrow \bR$ and define
$$f(X_t) = \bE\left[\left.e^{-r(T-t)} g(S_T) + \int_t^T e^{-r(u - t)} k(S_u) du \ \right| \ X_t \right],$$
for any path $X_t \in \Lambda$, $t \in [0,T]$. Thus, if $f \in \bC^{1,2}$ and $k$ is $\Lambda$-continuous, then $f$ satisfies the following Path-dependent Partial Differential Equation (PPDE):
\begin{align*}
\Delta_t f(X_t) + a(t, x_t) \Delta_x f(X_t)  + \frac{1}{2} b^2(t, x_t) \Delta_{xx} f(X_t) - r f(X_t) + k(X_t) = 0,
\end{align*}
with $f(X_T) = g(X_T)$, for any $X_T$ in the topological support of the process $s$.
\end{teo}

As we will observe, the commutation of the time and space functional derivatives plays an important role in the functional It\^o calculus theory, see, additionally, \cite{fito_saporito_greeks}. This will also be seen in the computation of the first-order approximation of path-dependent option prices. We will discuss the commutation issue in the next section.

\subsection{Weakly Path Dependent Functionals}\label{weakly_path_depend_section}

\begin{defi}[Lie Bracket]\label{lie_bracket}
The \textit{Lie bracket} of the operators $\Delta_t$ and $\Delta_x$ is defined as
$$[\Delta_x, \Delta_t] f(X_t) = \Delta_{xt} f(X_t) - \Delta_{tx} f(X_t),$$
where $\Delta_{tx} = \Delta_t \Delta_x$ and $f$ is such that all the derivatives above exist.
\end{defi}

It is a instantaneous measurement of the path-dependence of the functional $f$, i.e. it will be zero if, in the limit, makes no difference the order of the bump and the flat extension of the path, see Figure \ref{fig_bracket}.
\begin{figure}[!ht]
\begin{center}
 \scalebox{0.8}{\includegraphics{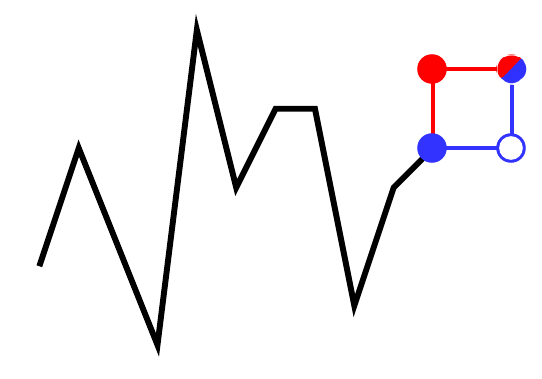}} \caption{Interpretation of the $[\Delta_x, \Delta_t].$} \label{fig_bracket}
\end{center}
\end{figure}

\begin{defi}[Locally Weakly Path Dependent]
A functional $f:\Lambda \longrightarrow \bR$ is called \textit{(locally) weakly path-dependent} if
$$[\Delta_x, \Delta_t] f = 0.$$
\end{defi}

Some examples of functionals should help understand these concepts.
\begin{enumerate}
\item $f(X_t) = h(t,x_t)$, with $h$ smooth. It is clearly weakly path-dependent (it actually is path-independent).
\item $f(X_t) = \ds\int_0^t x_u du$, the time integral of the current path. A simple computation shows
$$[\Delta_x, \Delta_t] f(X_t) = 1,$$
giving then an example of a functional which is not weakly path-dependent.
\item An example of a (locally) weakly path-dependent functional which is not path-independent is $f(X_t) = \ds\int_0^t \int_0^s x_u du ds$, because
$$\Delta_x f(X_t) = 0 \mbox{ and } \Delta_t f(X_t) = \int_0^t x_s ds.$$
\end{enumerate}

\subsection{Path-Dependent Derivative Pricing in the Black--Scholes Model}

The Black--Scholes model assumes that
$$ds_t = rs_t dt + \sigma s_t dw_t.$$
In this case, the topological support of $s$ is the set of continuous paths taking values in the positive real line. Hence, if we denote the price of a derivative with payoff $g$ under this model by $P_{BS}$, we find the path-dependent version of the Black-Scholes PPDE
\begin{align}\label{eq:black_scholes_ppde}
\Delta_t P_{BS}(X_t) + \frac{1}{2} \sigma^2 x_t^2 \Delta_{xx} P_{BS}(X_t) + r (x_t \Delta_x P_{BS}(X_t) - P_{BS}(X_t))  = 0,
\end{align}
for any path $X_t$ in the aforesaid set of paths, and with $P_{BS}(X_T) = g(X_T)$. Henceforth, we define, slightly abusing the notation, the Black--Scholes PPDE operator as
\begin{align}\label{eq:black_scholes_ppde_operator}
\cL_{BS}(\sigma) = \Delta_t + \frac{1}{2} \sigma^2 \bD_2 + r \bD_2  - r \cdot,
\end{align}
with
\begin{align}
\bD_k f(X_t) = x_t^k (\Delta_x)^k f(X_t). \label{eq:D_k_path}
\end{align}

\begin{lema}\label{lem:commutation}
Consider the following path-dependent operator:
$$\cA = \sum_{k=1}^n a_k \bD_k,$$
where $a_i \in \bR$. Then
$$[\Delta_t, \cA]f = 0 \Leftrightarrow [\cL_{BS}(\sigma), \cA]f = -r \cA f$$
\end{lema}

\begin{proof}
Notice that, since the operators $\bD_k$'s commute among themselves, $[\Delta_t, \cA]f = 0$ is equivalent to $[\cL_{BS}(\sigma), \cA]f = - r \cA f$.
\end{proof}

\begin{prop}\label{prop:sol_PPDE_source}
Define the operator $\cA$ as in lemma above:
$$\cA = \sum_{k=1}^n a_k \bD_k$$
and let $f_0$ be a functional that solves $\cL_{BS}(\sigma)f_0 = 0$. Consider then the PPDE
$$\left\{
\begin{array}{l}
 \cL_{BS}(\sigma) f(X_t) = \psi(t)\cA f_0(X_t), \\ \\
f(X_T) = 0,
\end{array}
\right. $$
for any continuous paths. If $[\Delta_t, \cA]f_0 = 0$ for continuous paths, then
$$f(X_t) = -\left(\int_t^T \psi(u) du\right) \cA f_0(X_t)$$
is a solution of the PPDE above.
\end{prop}

\begin{proof}
Define $\tilde{f}(X_t) = \phi(t) \cA f_0(X_t)$, with $\phi(T) = 0$, and notice that, clearly, $\tilde{f}(X_T) = 0$. Moreover, by Lemma \ref{lem:commutation},
\begin{align*}
\cL_{BS}(\sigma) \tilde{f}(X_t) &= (\phi'(t) - r \phi(t)) \cA f_0(X_t) + \phi(t) ( \cL_{BS}(\sigma) + r \cdot) \cA f_0(X_t)\\ 
&=  (\phi'(t) - r \phi(t)) \cA f_0(X_t) +\phi(t) \cA \cancelto{r f_0}{( \cL_{BS}(\sigma) + r \cdot) f_0(X_t)} \\
&+ \phi(t) \cancelto{0}{[\cL_{BS}(\sigma) + r  \cdot, \cA]f_0(X_t)} = \phi'(t) \cA f_0(X_t).
\end{align*}
Hence, solving $\phi'(t) = \psi(t)$ with $\phi(T) = 0$, we conclude the argument.
\end{proof}

\begin{obs}[Path-dependent Vega]\label{obs:vega}
The Vega of a path-dependent option under the Black-Scholes model with price $P_{BS}(X_t,\sigma)$ will be denoted by $\nu(X_t,\sigma) = \ds\partial P_{BS}/\partial \sigma(X_t,\sigma)$. Moreover, by the PPDE (\ref{eq:black_scholes_ppde}), $\nu$ solves the PPDE:
$$\left\{
\begin{array}{l}
  \ds \cL_{BS}(\sigma)\nu(X_t,\sigma) = -\sigma \bD_2 P_{BS}(X_t,\sigma),\\ \\
  \nu(X_T,\sigma) = 0,
\end{array}
\right.$$
where we have differentiated the Black--Scholes PPDE (\ref{eq:black_scholes_ppde}) with respect to the parameter $\sigma$. By the Functional Feynman-Kac Formula, Theorem \ref{feynman-kac}, $\nu$ might be represented as
$$\nu(X_t,\sigma) = \bE\left[ \left.\int_t^T e^{-r(u-t)} \sigma \bD_2 P_{BS}(S_u,\sigma) du \ \right| S_t = X_t \right],$$
where $s$ here follows the Black-Scholes SDE with volatility $\sigma$. Furthermore, if $[\Delta_t, \bD_2]P_{BS}(X_t,\sigma) = 0$, we have the well-known relation:
$$\nu(X_t,\sigma) = (T-t)\sigma \bD_2 P_{BS}(X_t,\sigma),$$
see Proposition \ref{prop:sol_PPDE_source}. The aforementioned commutation condition is verified, for instance, for path-independent option prices, giving us the well-known relation between the Gamma and the Vega.
\end{obs}

\subsection{Formal Derivation of the Functional First-Order Approximation}

Fix a maturity $T$ and a payoff functional $g:\Lambda_T \longrightarrow \bR$. Since $(s, y^\eps, z^\delta)$ is a Markovian process, the no-arbitrage price of this path-dependent European derivative depends on the realized path of $s$, but only at the spot values of $y^\eps$ and $z^\delta$. The only source of path-dependence is in the payoff $g$ and hence $s$ is the only variable the knowledge of its current path is necessary. Hence, the no-arbitrage price of this path-dependent European derivative under the model (\ref{eq:sde_risk_neutral_intro}) is given by
$$P^{\eps,\delta}(X_t,y,z) = \bE[e^{-r(T-t)} g(S_T) \ | \ S_t = X_t, y^{\eps}_t = y, z^{\delta}_t = z].$$

\begin{obs}
Under some mild conditions on the coefficients and the payoff function $g$, the functional $P^{\eps,\delta}$ belongs to $\bC^{1,2}$. Additionally, as stated in Remark \ref{obs:smoothness}, we will assume $P^{\eps, \delta}$ is as smooth as needed in the computations that follow.
\end{obs}

We denote the functional infinitesimal generator of $(s,y^{\eps},z^{\delta})$ by $\cL_{(s,y,z)}^{\eps,\delta}$ and write $\cL^{\eps,\delta}$ as:
\begin{align}
\cL^{\eps,\delta} &= \Delta_t + \cL_{(s,y,z)}^{\eps,\delta} - r \cdot \label{eq:pricing_PPDE}\\
&= \frac{1}{\eps} \cL_0 + \frac{1}{\sqrt{\eps}} \cL_1 + \cL_2 + \sqrt{\delta} \cM_1 + \delta \cM_2 + \sqrt{\frac{\delta}{\eps}} \cM_3, \nonumber
\end{align}
where
\begin{align}
\cL_0 &= \alpha(y) \frac{\partial}{\partial y} + \frac{1}{2} \beta^2(y) \frac{\partial^2}{\partial y^2}, \label{eq:path_cL_0} \\
\cL_1 &= \beta(y) \left( \rho_1 f(y,z) \frac{\partial}{\partial y} \bD_1 -\Gamma_1(y,z) \frac{\partial}{\partial y}\right), \label{eq:path_cL_1} \\
\cL_2 &= \Delta_t + \frac{1}{2} f^2(y,z) \bD_2 + r \bD_1 - r \cdot, \label{eq:path_cL_2} \\
\cM_1 &= g(z) \left(\rho_2 f(y,z) \frac{\partial}{\partial z}\bD_1 - \Gamma_2(y,z) \frac{\partial}{\partial z} \right), \label{eq:path_cM_1} \\
\cM_2 &= c(z) \frac{\partial}{\partial z} + \frac{1}{2} g^2(z) \frac{\partial^2}{\partial z^2}, \label{eq:path_cM_2} \\
\cM_3 &= \rho_{12} \beta(y)g(z) \frac{\partial^2}{\partial y \partial z}, \label{eq:path_cM_3}
\end{align}
Hence, the Functional Feynman-Kac Formula, Theorem \ref{feynman-kac}, implies that $P^{\eps, \delta}$ satisfies the the following PPDE
\begin{align}\label{eq:ppde_S}
\left\{
\begin{array}{l}
 \cL^{\eps,\delta}P^{\eps,\delta}(X_t,y,z) = 0, \\ \\
 P^{\eps,\delta}(X_T,y,z) = g(X_T).
\end{array}
\right.
\end{align}

The functional differential operators (\ref{eq:path_cL_0})--(\ref{eq:path_cM_3}) can also be described in words, which shall help the reader understand how the elements of the model work separately in the PPDE (\ref{eq:pricing_PPDE}):
\begin{itemize}

\item $\cL_0$ is the infinitesimal generator of $y^1$ under $\bP$;

\item $\cL_1$ is composed by a term due to the correlation of the stock price and the fast factor, and a term due to part of the market price of volatility risk;

\item $\cL_2$ is the path-dependent Black--Scholes operator with volatility $f(y,z)$;

\item $\cM_1$ is composed by a term due to the correlation of the stock price and the slow factor;

\item $\cM_2$ is the infinitesimal generator of $z^1$ under $\bP$;

\item $\cL_3$ is composed by a unique term due to the correlation of fast and slow factors.

\end{itemize}

These are virtually the same operators as in the classical case. The only difference is the presence of the functional derivatives $\Delta_t$ and $\Delta_x$ instead of the partial derivatives $\partial/\partial t$ and $\partial/\partial x$. Notice that $\cL_0$, $\cM_2$ and $\cM_3$ are standard differential operators and $\cL_1$, $\cL_2$ and $\cM_1$ are functional differential operators.

\begin{obs}[Commutation]

It is important to notice that although $\Delta_t$ and $\Delta_x$ do not commute (see Section \ref{weakly_path_depend_section}), $\partial/\partial y$ and $\partial/\partial z$ commute with both $\Delta_t$ and $\Delta_x$.

\end{obs}

We will now formally derive the first-order approximation $P^{\eps,\delta}$. Write $P^{\eps,\delta}$ in powers of $\sqrt{\delta}$:
$$P^{\eps,\delta} = P^{\eps}_0 + \sqrt{\delta} P^{\eps}_1 + \cdots,$$
and then we choose $P^{\eps}_0$ and $P^{\eps}_1$ to satisfy

\begin{align}
&\left\{
\begin{array}{l}
 \ds \left(\frac{1}{\eps}\cL_0 + \frac{1}{\sqrt{\eps}} \cL_1 + \cL_2\right)P^{\eps}_0(X_t,y,z)= 0, \\ \\
 P^{\eps}_0(X_T,y,z) = g(X_T),
\end{array}
\right. \label{eq:path_pde_p0eps} \\ \nonumber\\
&\left\{
\begin{array}{l}
 \ds\left(\frac{1}{\eps} \cL_0 + \frac{1}{\sqrt{\eps}} \cL_1 + \cL_2\right)P^{\eps}_1 = -\left(\cM_1 + \frac{1}{\sqrt{\eps}} \cM_3 \right)P_0^{\eps},\\ \\
 P^{\eps}_1(X_T,y,z) = 0.
\end{array}
\right. \label{eq:path_pde_p_1eps}\\ \nonumber
\end{align}

\subsubsection{Computing $P_0$}

Expand $P^{\eps}_0$ in powers of $\sqrt{\eps}$:
$$P_0^{\eps} = \sum_{m \geq 0} (\sqrt{\eps})^m P_{m,0},$$
where we denote $P_{0,0}$ by $P_0$. Substitute now this expansion into Equation (\ref{eq:path_pde_p0eps}) to get the following PPDEs
$$\begin{array}{rl}
(-1,0):& \cL_0P_0 = 0,\\
(-1/2,0):& \cL_0P_{1,0}+ \cL_1P_0 = 0,\\
(0,0):& \cL_0P_{2,0}+ \cL_1P_{1,0} + \cL_2P_0= 0,\\
(1/2,0):& \cL_0P_{3,0} + \cL_1 P_{2,0}+ \cL_2P_{1,0} = 0,
\end{array}$$
where we are using the notation $(i,j)$ to denote the term of $i$th order in $\eps$ and $j$th in $\delta$.

It will be paramount in the computations that follows to notice that $\cL_0$ is an usual differential operator. Hence, the first Equation above is actually a PDE and the arguments $(X_t,z)$ should be understood as parameters in this equation. The reasoning used in this section follows the steps of the classical case described in \cite{multiscale_fouque_new_book}.

We take $P_0 = P_0(X_t,z)$ independent of $y$ in order to satisfy the first PDE. Since $\cL_1$ takes derivative with respect to $y$ in all its terms, $\cL_1P_0 = 0$ and then the second equation becomes the PDE $\cL_0 P_{1,0} = 0$. Thus, we take $P_{1,0} = P_{1,0}(X_t,z)$ also independent of $y$. The 0-order term gives us
\begin{align}\label{eq:poisson_p_20}
\cL_0P_{2,0} + \cancelto{\scriptstyle 0}{\cL_1 P_{1,0}} + \cL_2P_0 = 0,
\end{align}
which is a Poisson PDE for $P_{2,0}$ with solvability condition:
$$\langle \cL_2P_0 \rangle = 0,$$
where $\langle \cdot \rangle$ is the average under the invariant measure of $\cL_0$. Note again that $(X_t,z)$ is seen as parameters here. Since $P_0$ does not depend on $y$, the solvability condition becomes
$$\langle \cL_2 \rangle P_0 = 0.$$
Again, we would like to point it out again that $\cL_0$ is a differential operator in the classical sense, since no functional derivatives are present. Hence, all the results regarding the Poisson PDE hold. 

Using the definition of the path-dependent Black--Scholes differential operator given in Equation (\ref{eq:black_scholes_ppde_operator}), we shall choose $P_0$ to satisfy the PPDE
\begin{align}
\left\{
\begin{array}{l}
 \cL_{BS}(\overline{\sigma}(z)) P_0(X_t,z) = 0, \\ \\
 P_0(X_T,z) = g(X_T),
\end{array}
\right. \label{eq:P0_path}
\end{align}
where $\overline{\sigma}^2(z) = \langle f^2(\cdot,z) \rangle$. Notice we can write $P_0(X_t,z) = P_{BS}(X_t;\overline{\sigma}(z))$, i.e. $P_0$ is the price of the path-dependent option with payoff $g$ and maturity $T$ under the Black-Scholes model:
\begin{align}\label{eq:black_scholes_sigma_bar_path}
d\overline{s}_t = r\overline{s}_t dt + \overline{\sigma}(z) \overline{s}_t d\overline{w}_t^{(0)}.
\end{align}

\subsubsection{Computing $P^{\eps}_{1,0}$}

By the Poisson PDE (\ref{eq:poisson_p_20}), we find
\begin{align}\label{eq:P20_path}
P_{2,0}(X_t,y,z) = -\cL_0^{-1}(\cL_{BS}(f(y,z)) - \cL_{BS}(\overline{\sigma}(z)))P_0(X_t,z) + c(X_t,z),
\end{align}
for some functional $c$ which does not depend on $y$. Notice
$$\cL_{BS}(f(y,z)) - \cL_{BS}(\overline{\sigma}(z)) = \frac{1}{2}(f^2(y,z) - \overline{\sigma}^2(z)) \bD_2.$$
Denote now by $\phi(y,z)$ the solution of the Poisson equation
\begin{align}\label{eq:path_poisson_phi}
\cL_0 \phi(y,z) = f^2(y,z) - \overline{\sigma}^2(z),
\end{align}
which implies
\begin{align*}
\cL_0^{-1}(\cL_{BS}(f(y,z)) - \cL_{BS}(\overline{\sigma}(z))) &= \frac{1}{2} \phi(y,z) \bD_2.
\end{align*}
Using the $1/2$-order PPDE, we get then the solvability condition
$$\langle \cL_1 P_{2,0} + \cL_2 P_{1,0}\rangle = 0.$$
Hence, by the Equation (\ref{eq:P20_path}) for $P_{2,0}$,
\begin{align*}
\cL_1 P_{2,0} &= - \cL_1 \left( \cL_0^{-1}(\cL_{BS}(f(y,z)) - \cL_{BS}(\overline{\sigma}(z)))P_{S_0}\right) \\
&= -\cL_1 \left(\frac{1}{2} \phi(y,z) \bD_2 \right)P_0 \\
&= \beta(y) \left( \rho_1 f(y,z) \frac{\partial}{\partial y} \bD_1 -\Gamma_1(y,z) \frac{\partial}{\partial y}\right)\frac{1}{2} \phi(y,z)\bD_2 P_0 \\
&= \left(\frac{1}{2}  \rho_1 \beta(y)f(y,z) \frac{\partial\phi}{\partial y}(y,z)\right) \bD_1 \bD_2 - \left(\frac{1}{2}\beta(y)\Gamma_1(y,z) \frac{\partial\phi}{\partial y}(y,z)\right) \bD_2P_0.
\end{align*}
Therefore, $P_{1,0}^{\eps} = \sqrt{\eps}P_{1,0}$ satisfies the following PPDE:
$$\left\{
\begin{array}{l}
 \ds \cL_{BS}(\overline{\sigma}(z))P^{\eps}_{1,0}(X_t,z) = -\cA^{\eps} P_0(X_t,z),\\ \\
 P^{\eps}_{1,0}(X_T,z) = 0,
\end{array}
\right.$$
where
\begin{align}
\cA^{\eps} &= V_3^{\eps}(z) \bD_1 \bD_2 + V_2^{\eps}(z) \bD_2, \label{eq:path_Aeps} \\
V_3^{\eps}(z) &= - \frac{\rho_1\sqrt{\eps}}{2} \left\langle \beta \ f(\cdot,z) \frac{\partial\phi}{\partial y}(\cdot,z) \right\rangle, \label{eq:path_V3eps} \\
V_2^{\eps}(z) &= \frac{\sqrt{\eps}}{2} \left\langle \beta \ \Gamma_1(\cdot,z) \frac{\partial\phi}{\partial y}(\cdot,z) \right\rangle. \label{eq:path_V2eps}
\end{align}
The only difference between this PPDE and the PDE of the classical case, Equation (\ref{eq:pde_p_1_eps_intro}), is that $\bD_k$ now involves the functional space derivative $\Delta_x$ and $\cL_{BS}$ is the functional version of the Black-Scholes differential operator.

\begin{obs}\label{obs:D1D2_D3}
Notice that $\bD_1 \bD_2 = 2 \bD_2 + \bD_3$ and we may rewrite
$$\cA^{\eps} = V_3^{\eps}(z)\bD_3 + (V_2^{\eps}(z) + 2V_3^{\eps}(z)) \bD_2.$$
\end{obs}

By the Functional Feynman-Kac Formula, Theorem \ref{feynman-kac}, we can write
\begin{align}
P^{\eps}_{1,0}(X_t,z) = \bE\left[ \left.\int_t^T e^{-r(u-t)} \cA^{\eps} P_0(\overline{S}_u,z) du \ \right| \ \overline{S}_t = X_t \right],\label{eq:P10_eps_path}
\end{align}
where $\overline{s}$ follows the Black--Scholes dynamics with volatility $\overline{\sigma}(z)$ as in (\ref{eq:black_scholes_sigma_bar_path}).

It is very important to notice that $V_2^{\eps}(z)$ and $V_3^{\eps}(z)$ are the same constants as in the path-independent case described in Section \ref{sec:first_order}. This means that once these parameters are calibrated to European vanilla options, we could use these same numeric values to price path-dependent options. The same is true for the $P^{\delta}_{0,1}$, which will be shown next section.

\subsubsection{Computing $P^{\delta}_{0,1}$}

Let us now expand $P_1^{\eps}$ in powers of $\sqrt{\eps}$,
$$P_1^{\eps} = \sum_{m \geq 0} (\sqrt{\eps})^m P_{m,1},$$
and then substitute this and the expansion for $P_0^{\eps}$ into Equation (\ref{eq:path_pde_p_1eps}). Doing so, we find
$$\begin{array}{rl}
(-1,-1/2):& \cL_0 P_{0,1} = 0,\\
(-1/2,-1/2):& \cL_0 P_{1,1}+ \cL_1 P_{0,1} + \cM_3 P_0 = 0,\\
(0,-1/2):& \cL_0 P_{2,1}+ \cL_1 P_{1,1} + \cL_2 P_{0,1} + \cM_1 P_0 + \cM_3P_{1,0} = 0.
\end{array}$$
Note that all the terms in $\cL_0$, $\cL_1$ and $\cM_3$ take derivative with respect to $y$. So, the first PDE is satisfied if we choose $P_{0,1} = P_{0,1}(X_t,z)$, as it was done previously. Now, the second PPDE turns to be the PDE $\cL_0P_{1,1} = 0$, and then we choose $P_{1,1} = P_{1,1}(X_t,z)$ independent of $y$ as well. Finally, the last PPDE becomes
$$\cL_0 P_{2,1} + \cL_2 P_{0,1} + \cM_1 P_0 = 0,$$
which is a Poisson PDE for $P_{2,1}$, and its solvability condition is
$$\langle \cL_2 P_{0,1} + \cM_1 P_0 \rangle = 0.$$
Thus, if we write $P^{\delta}_{0,1}(X_t,z) = \sqrt{\delta}P_{0,1}(X_t,z)$, this condition can be written as
$$ \ds \cL_{BS}(\overline{\sigma}(z)) P^{\delta}_{0,1} = -\sqrt{\delta} \langle \cM_1 \rangle P_0,$$
where one can compute
$$\langle \cM_1 \rangle = -g(z) \langle \Gamma_2(\cdot,z) \rangle \frac{\partial}{\partial z} + \rho_2 g(z) \langle f(\cdot,z) \rangle \bD_1 \frac{\partial}{\partial z}.$$
Therefore, if we define
\begin{align}
\cA^{\delta} &= - V_0^{\delta}(z) \frac{\partial}{\partial \sigma}  -V_1^{\delta}(z) \bD_1\frac{\partial}{\partial \sigma} \label{eq:path_Adelta} \\
V_1^{\delta} &= \frac{\rho_2g(z)\sqrt{\delta}}{2} \langle f(\cdot,z) \rangle \overline{\sigma}'(z) \label{V1delta} \\
V_0^{\delta} &= -\frac{g(z)\sqrt{\delta}}{2} \langle \Gamma_2(\cdot,z) \rangle \overline{\sigma}'(z) \label{V0delta}
\end{align}
we have the following PPDE:
$$\left\{
\begin{array}{l}
 \ds \cL_{BS}(\overline{\sigma}(z)) P^{\delta}_{0,1}(X_t,z) = -2 \cA^{\delta} P_0(X_t,z),\\ \\
 P^{\delta}_{0,1}(X_T,z)= 0.
\end{array}
\right.$$
In general, by the Functional Feynman-Kac's Formula, Theorem \ref{feynman-kac},
\begin{align}
P^{\delta}_{0,1}(X_t,z) = 2 \ \bE\left[ \left.\int_t^T e^{-r(u-t)} \cA^{\delta} P_0(\overline{S}_u,z) du \ \right| \ \overline{S}_t = X_t \right]. \label{eq:P01_delta_path}
\end{align}

\begin{obs}[Parameter Reduction]
As in the path-independent case, parameter reduction could still be performed and therefore we can restrict ourselves to the group market parameters:
$$\{\sigma^{\star}(z), V_0^{\delta}(z), V_1^{\delta}(z), V_3^{\eps}(z)\}.$$
\end{obs}

\subsection{Asian Options}

To exemplify the result above, let us consider the case where the contract functional $g$ is of the form $g(X_T) = \varphi(x_T, I(X_T))$, where
$$I(X_t) = \int_0^t x_u du.$$
See, for instance, \cite{multiscale_fouque_asian}. In this case, $P_0(X_t, z) = \varphi_0(t, x_t, I(X_t), \overline{\sigma}(z))$. By Equation (\ref{eq:P0_path}) and since $\Delta_t I(X_t) = x_t$, it is clear to see that $\varphi_0$ satisfies the usual pricing PDE for Asian options under the Black--Scholes:
$$\left\{
\begin{array}{l}
\ds \frac{\partial \varphi_0}{\partial t} + x \frac{\partial \varphi_0}{\partial I} + r x \frac{\partial \varphi_0}{\partial x} + \frac{1}{2} \sigma^2 x^2 \frac{\partial^2 \varphi_0}{\partial x^2} - r \varphi_0 = 0, \\ \\
\varphi_0(T, x, I,  \overline{\sigma}(z)) = \varphi(x, I).
\end{array}
\right.$$

Moreover, by Equations (\ref{eq:P10_eps_path}), (\ref{eq:P01_delta_path}) and by the Functional Feynman-Kac Formula, Theorem \ref{feynman-kac}, we find the 
\begin{align*}
\varphi^{\eps}_{1,0}(t, x, I, \overline{\sigma}(z)) &= \bE\left[ \left.\int_t^T e^{-r(u-t)} \cA^{\eps} \varphi_0(u, \overline{s}_u, I(\overline{S}_u), \overline{\sigma}(z)) du \ \right| \ \overline{s}_t = x, I(\overline{S}_t) = I\right],\\
\varphi^{\delta}_{0,1}(t, x, I, \overline{\sigma}(z)) &= 2 \ \bE\left[ \left.\int_t^T e^{-r(u-t)} \cA^{\delta} \varphi_0(u, \overline{s}_u, I(\overline{S}_u), \overline{\sigma}(z)) du \ \right| \ \overline{s}_t = x, I(\overline{S}_t) = I \right].
\end{align*}
Furthermore, since $\Delta_x I(X_t) = 0$, we have
$$\bD_k \varphi_0(t, x_t, I(X_t),  \overline{\sigma}(z)) = x_t^k \frac{\partial^k \varphi_0}{\partial x^k}(t, x_t, I(X_t),  \overline{\sigma}(z)),$$
and the Vega of $\varphi_0$ could be numerically computed by using the expression delineated in Remark \ref{obs:vega}. Therefore, the first-order approximation could be numerically calculated using the equations above.

\subsection{Closed-form Solutions in the \textit{Weaker} Path-dependent Case}

We will now prove that the formula for $P^{\eps}_{1,0}$ and $P_{0,1}^\delta$ we presented in Section \ref{sec:first_order} is valid here as well, as long as we assume the path dependence structure of the price $P_0$ is not very strong. We will make these claims precise now. 

\begin{ass}\label{ass:weaker_path_dependence}
For every continuous path $X_t$,
$$[\Delta_t, \Delta_{xx}]P_0(X_t,z) = [\Delta_t, \Delta_{xxx}]P_0(X_t,z) = 0.$$
\end{ass}

\begin{prop}\label{prop:simple_representation}

If the zero-order price $P_0$ satisfies Assumption \ref{ass:weaker_path_dependence}, then we find the well-known formula
\begin{align}
P^{\eps}_{1,0}(X_t,z) &= (T-t) \cA^{\eps} P_0(X_t,z), \label{eq:p_eps_formula}\\
P^{\delta}_{0,1}(X_t,z) &= (T-t) \cA^{\delta} P_0(X_t,z),\label{eq:p_delta_formula}
\end{align}
for any continuous path. More directly, by Remark \ref{obs:D1D2_D3}, the first-correction is given by
\begin{align}\label{eq:correction_weak}
&(T-t)\left(V_3^{\eps}(z) - \overline{\sigma}(z)V_1^{\delta}(z) \right) \bD_3 P_0 \\
&+ (T-t)\left(V_2^{\eps}(z) + 2V_3^{\eps}(z) - \overline{\sigma}(z)(V_0^{\delta}(z) + 2V_1^{\delta}(z))  \right) \bD_2 P_0 \nonumber
\end{align}

\end{prop}

\begin{proof}
Notice that Assumption \ref{ass:weaker_path_dependence} implies that $[\Delta_t, \cA^{\eps}]P_0 = 0$ and then Equation (\ref{eq:p_eps_formula}) follows directly from Proposition \ref{prop:sol_PPDE_source}. The commutation requirement $[\Delta_t, \cA^{\delta}]P_0 = 0$ is not readily related to commutation of $\Delta_t$ and $\Delta_{xx}$ and $\Delta_{xxx}$, as it is in the fast mean-reverting case. However, as we have seen in Remark \ref{obs:vega} on the relation of the functional Vega and Gamma, since $[\Delta_t, \Delta_{xx}]P_0 = 0$, then
$$\cA^{\delta} P_0 = -(T-t) \overline{\sigma}(z)V_1^{\delta}(z)  \bD_1\bD_2 P_0 - (T-t) \overline{\sigma}(z) V_0^{\delta}(z) \bD_2 P_0.$$
Hence, since we also have $[\Delta_t, \Delta_{xxx}]P_0 = 0$, we conclude that $[\Delta_t, \cA^{\delta}]P_0 = 0$. Therefore, by Proposition \ref{prop:sol_PPDE_source}, we have the desired result.
\end{proof}

\begin{obs}\label{obs:assumption_weakly_pd}
If $P_0$ is weakly path-dependent (i.e. $[\Delta_t, \Delta_x] P_0 = 0$), one can straightforwardly show that Assumption \ref{ass:weaker_path_dependence} is equivalent to $\Delta_x P_0$ and $\Delta_{xx} P_0$ being weakly path-dependent as well.
\end{obs}

\subsection{Option on Quadratic Variation}

We will consider an option with payoff $g(X_T) = \varphi(x_T, QV(X_T))$, where $QV$ is the quadratic variation functional. We forward the reader to \cite{fito_saporito_greeks} for the details on this type of options and its properties, including the pathwise definition of the quadratic variation functional. We write $P_0(X_t, z) = \varphi_0(t, x_t, QV(X_t), z)$ and using the fact that $\Delta_t QV(X_t) = 0$, $\Delta_x QV(X_t) = 2(x_t - x_{t-})$ and $\Delta_{xx} QV(X_t) = 2$, we can readily show that $[\Delta_t, \Delta_x]P_0 = 0$, for continuous paths, and hence it is weakly path-dependent. Moreover, 
\begin{align*}
\Delta_x P_0(X_t,z) &= \frac{\partial \varphi_0}{\partial x},\\
\Delta_{xx} P_0(X_t, z) &= \frac{\partial^2 \varphi_0}{\partial x^2} + 2 \frac{\partial \varphi_0}{\partial QV},\\
\Delta_{xxx} P_0(X_t, z) &= \frac{\partial^3 \varphi_0}{\partial x^3} + 6 \frac{\partial^2 \varphi_0}{\partial x \partial QV},
\end{align*}
for every continuous path $X_t$. Hence, by Remark \ref{obs:assumption_weakly_pd}, $P_0$ satisfies Assumption \ref{ass:weaker_path_dependence}. These formulas can be applied to computationally find the first-order correction as outlined in Proposition \ref{prop:simple_representation}.

\subsection{Accuracy Theorem}

\begin{teo}\label{thm:accuracy_path}

We assume items 1 to 5 from Theorem \ref{thm:accuracy_intro} and additionally that

\begin{enumerate}

\item[(6*)] The payoff functional $g$ is such that $P^{\eps,\delta}$ is smooth as in Remark \ref{obs:smoothness}.

\end{enumerate}
Then,
$$P^{\eps,\delta}(X_t,y,z) = P_0(X_t,z) + P_{1,0}^{\eps}(X_t,z) + P_{0,1}^{\delta}(X_t,z) + O(\eps + \delta),$$
with $P_0$, $P_{1,0}^{\eps}$ and $P_{0,1}^{\delta}$ given by Equations (\ref{eq:P0_path}), (\ref{eq:P10_eps_path}) and (\ref{eq:P01_delta_path}), respectively.
\end{teo}

For the attentive reader, it should be clear by now the similarities of the first-order perturbation method in the classical and functional frameworks. Hence, it should be also clear that the same proof of the accuracy of the functional first-order approximation can be carried out without much difficulty. In fact, the definition of the higher-order approximation of $P^{\eps,\delta}$ and the analysis of residual of such approximation follows identically to the classical case. Since the Feynman-Kac formula is also available in the functional framework, see Theorem \ref{feynman-kac}, the study of the boundedness of the residual follows similarly. The reader should also notice, as it was commented before, the properties of the processes $y^{\eps}$ and $z^{\delta}$ stay unadulterated since the functional aspect is considered only for the stock price variable. Furthermore, one should be able to consider weaker assumptions on the payoff functional $g$, which should be similar to the ones in the path-independent case stated in Theorem \ref{thm:accuracy_intro}, so that theorem above also holds true. We leave this for future work.

One of the main conclusions of this paper is the following corollary of the above theorem:

\begin{cor}
  
The market group parameters $(\overline{\sigma}(z), V_0^{\delta}(z), V_1^{\delta}(z), V_2^{\eps}(z),$ $V_3^{\eps}(z))$, or in their reduced form $\{\sigma^{\star}(z), V_0^{\delta}(z), V_1^{\delta}(z), V_3^{\eps}(z)\}$, do not change based on the path-dependence of the payoff functional. 
  
\end{cor}

\section{Conclusion and Future Work}

The main conclusion of this paper is that the first-order approximation for path-dependent options depend on the same market group parameters of the first-order approximation for vanilla options. Therefore, once the market group parameters are calibrated to vanilla option market data, one could use them to compute consistent first-order approximation of prices for path-dependent, exotic derivatives using the general representations (\ref{eq:P10_eps_path}) and (\ref{eq:P01_delta_path}). Without the functional It\^o calculus framework, the results above had to be stated and proved for each particular type of path-dependence. 

Moreover, when the path-dependence is not too strong, the first-order approximation we find for path-independent derivatives contracts, Equations (\ref{eq:p0_into})--(\ref{eq:p01delta_into}), holds. Namely, if the zero-order price satisfy the following commutation relations
$$[\Delta_t,\Delta_{xx}]P_0 = [\Delta_t,\Delta_{xxx}]P_0 = 0,$$
we have, essentially, the same formulas as in the path-independent, see Proposition \ref{prop:simple_representation}. However, the Greeks that compose this first-order correction are the path-dependent Greeks as in (\ref{eq:path_Aeps}) and (\ref{eq:path_Adelta}).

The development of numerical methods for the efficient computation of (\ref{eq:P10_eps_path}) and (\ref{eq:P01_delta_path}) will be left for future work. Additionally, forthcoming research will be conducted to weaken the smoothness assumption of Theorem \ref{thm:accuracy_path} in other to consider other types of contract functionals, as, for instance, barrier options, since the running maximum and minimum functionals are not smooth.

\subsubsection*{Acknowledgments}

I would like to thank J.-P. Fouque and B. Dupire for all the insightful discussions.

\bibliographystyle{plainnat}

\end{document}